\documentclass[conference]{IEEEtran}
\usepackage{times}
\usepackage[dvipsnames]{color, colortbl, xcolor}
\usepackage{multicol}
\usepackage[bookmarks=true,hidelinks]{hyperref}
\usepackage{amsfonts}
\usepackage{amsmath}
\usepackage{amsthm}
\usepackage{dsfont}
\usepackage{mathtools}
\usepackage{multicol}
\usepackage[pdftex]{graphicx}   
\usepackage{graphicx}
\usepackage{amssymb}

\usepackage{rotating}
\usepackage{makecell, multirow}

\usepackage[skip=1ex]{caption}

\usepackage[capitalize,noabbrev,nameinlink]{cleveref}
\crefformat{equation}{(#2#1#3)}
\usepackage[acronyms, shortcuts]{glossaries}
\usepackage[sort,numbers]{natbib}
\usepackage{cases}
\usepackage{xspace}
\usepackage{booktabs}
\usepackage{siunitx}
\usepackage{float}

\usepackage{subfigure}
\glsdisablehyper

\newacronym{ioc}{IOC}{inverse optimal control}
\newacronym{lqr}{LQR}{linear-quadratic regulator}
\newacronym{kkt}{KKT}{Karush–Kuhn–Tucker}
\newacronym{irl}{IRL}{inverse reinforcement learning}
\newacronym{mle}{MLE}{maximum likelihood estimation}
\newacronym[longplural=open-loop Nash equilibria,plural=OLNE]{olne}{OLNE}{open-loop Nash equilibrium}
\newacronym[longplural={partially observable Markov decision processes}]{pomdp}{POMDP}{partially observable Markov decision process}
\newacronym{svo}{SVO}{social value orientation}
\newacronym{ukf}{UKF}{unscented Kalman filter}
\newacronym{ibr}{IBR}{iterated best response}
\newacronym{awgn}{AWGN}{additive white Gaussian noise}
\newacronym{iqr}{IQR}{interquartile range}

\newcommand{\ie}{i.e.\@\xspace}

\newcommand{\cf}{c.f.\@\xspace}

\newcommand{\mbb}{\mathbb}
\newcommand{\mc}{\mathcal}

\newcommand{\R}{\mbb{R}}

\newcommand{\state}{x}
\newcommand{\bx}{{\state}}
\newcommand{\control}{u}
\newcommand{\bu}{{\control}}

\newcommand{\example}[1]%
{
\vspace{0.15cm}
\noindent \textit{\textbf{Running example:} #1}
\vspace{0.15cm}
}

\newtheorem{theorem}{Theorem}

\DeclareMathOperator*{\argmin}{arg\,min}

\DeclareGraphicsExtensions{.pdf,.png}  

\usepackage{ifthen}
\newboolean{include-notes}
\setboolean{include-notes}{true}

\makeatletter
\def\NAT@spacechar{~}
\makeatother

\usepackage{lipsum}

\newcommand{\makiescale}{0.5}
\newcommand{\makiepngscale}{0.19}

\begin{document}
\title{Learning Mixed Strategies in Trajectory Games}
\author{
\authorblockN{
Lasse Peters\authorrefmark{1}\quad
David Fridovich-Keil\authorrefmark{2}\quad
Laura Ferranti\authorrefmark{1}\quad
Cyrill Stachniss\authorrefmark{3}\quad
Javier Alonso-Mora\authorrefmark{1}\quad
Forrest Laine\authorrefmark{4}
}
\authorblockN{\\
\authorrefmark{1}
Delft University of Technology, Netherlands
\quad
\authorrefmark{2}
University of Texas at Austin, USA\\
\authorrefmark{3}
University of Bonn, Germany
\quad
\authorrefmark{4}
Vanderbilt University, USA
}
\authorblockA{
{
\texttt{
\{
\href{mailto://l.peters@tudelft.nl}{l.peters},
\href{mailto://l.ferranti@tudelft.nl}{l.ferranti},
\href{mailto://j.alonsomora@tudelft.nl}{j.alonsomora}
\}@tudelft.nl,}}\\
{\texttt{\href{mailto://dfk@utexas.edu}{dfk@utexas.edu},}}
{\texttt{\href{mailto://cyrill.stachniss@igg.uni-bonn.de}{cyrill.stachniss@igg.uni-bonn.de},}}
{\texttt{\href{mailto://forrest.laine@vanderbilt.edu}{forrest.laine@vanderbilt.edu}}}
}
}

\maketitle

\begin{abstract}
In multi-agent settings, game theory is a natural framework for describing the strategic interactions of agents whose objectives depend upon one another's behavior.
Trajectory games capture these complex effects by design.
In competitive settings, this makes them a more faithful interaction model than traditional ``predict then plan'' approaches.
However, current game-theoretic planning methods have important limitations. 
In this work, we propose two main contributions.
First, we introduce an offline training phase which reduces the online computational burden of solving trajectory games.
Second, we formulate a \emph{lifted} game which allows players to optimize multiple candidate trajectories in unison and thereby construct more competitive ``mixed'' strategies.
We validate our approach on a number of experiments using the pursuit-evasion game ``tag.''
\end{abstract}
\IEEEpeerreviewmaketitle

\section{Introduction}
\label{sec:intro}

Trajectory optimization techniques have become increasingly common in motion planning.
So long as vehicle dynamics, design objectives, and safety constraints satisfy mild regularity conditions, a motion planning problem may be encoded as a nonlinear program and solved efficiently to a locally-optimal solution.
The widespread successes of trajectory optimization have sparked growing interest in similar techniques for multi-agent, noncooperative decision-making and motion planning.
In this context, game theory offers an elegant mathematical framework for modeling the strategic interactions of rational agents with distinct interests.
By reasoning about interactions with others as a \emph{trajectory game}, an autonomous agent can plan future decisions while accounting for the strategic reactions of others.

Since they involve multiple players with distinct, potentially competing objectives, trajectory games can be far more complex to solve than single-agent trajectory optimization problems.
Recent algorithmic advances make solving trajectory games tractable in some instances \cite{fridovich2020icra, di2019cdc}.
Still, they remain fundamentally more challenging to solve than single-agent problems, and consequently, trajectory games have not been widely adopted in the robotics community.

\begin{figure}[ht!]
    \centering
    \subfigure[Pure strategies\label{fig:equilibrium-differences-pure}]{
    \centering
    \includegraphics[scale=\makiescale]{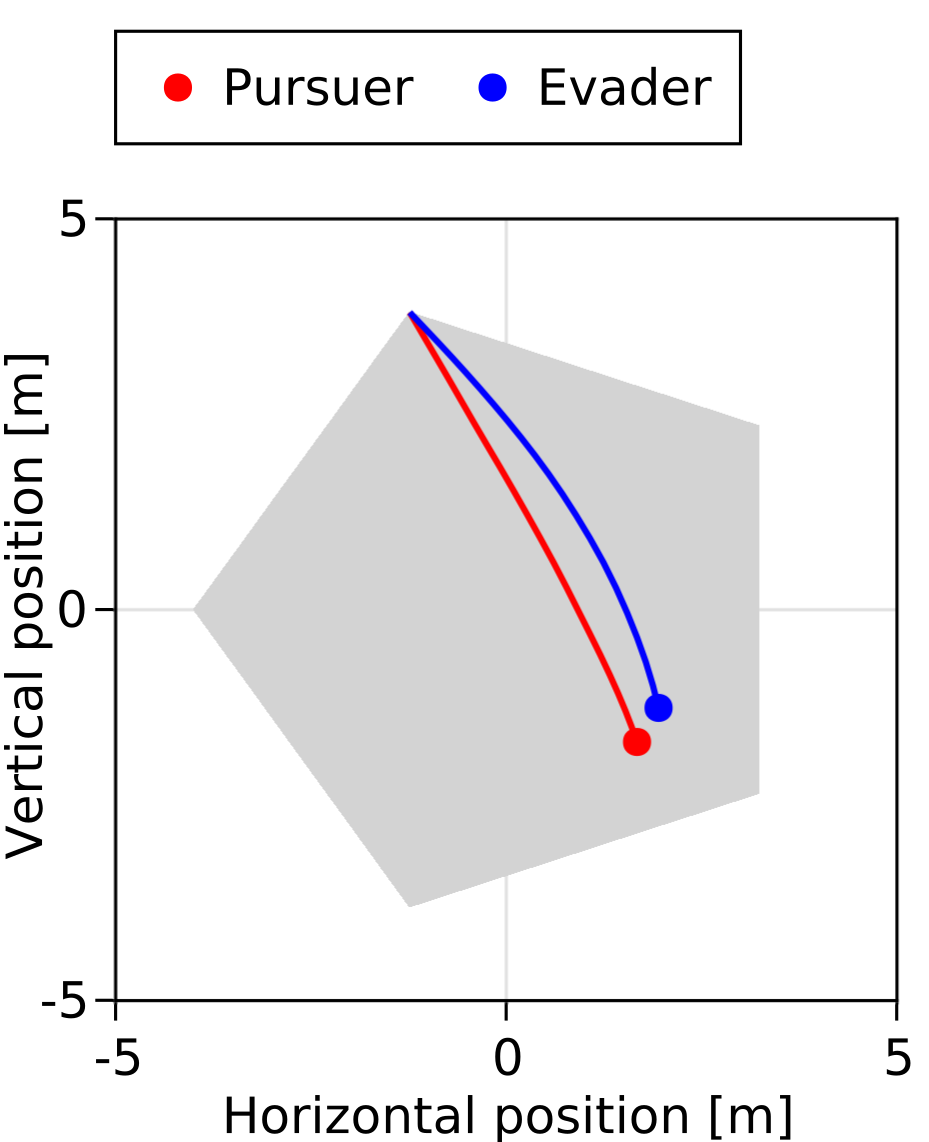}
    }
    \subfigure[Mixed strategies\label{fig:equilibrium-differences-mixed}]{
    \centering
    \includegraphics[scale=\makiescale]{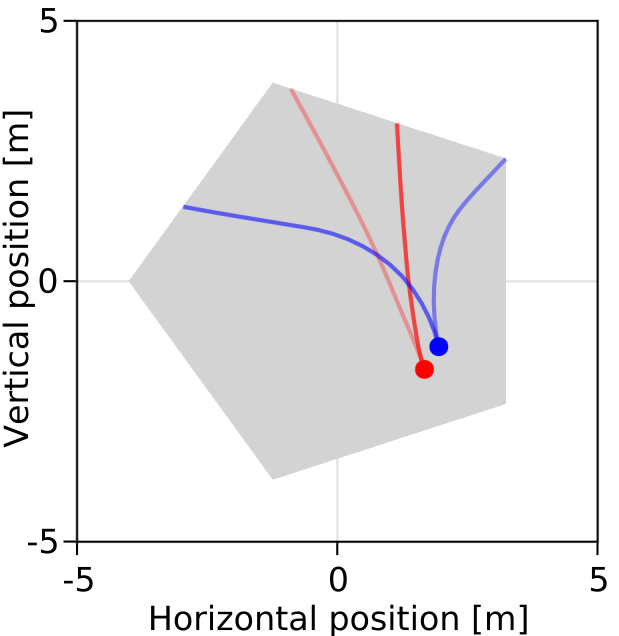}
    }
    \caption{
    A zero-sum game of tag played between two agents with planer point-mass dynamics in a pentagonal environment.
    (a)~In pure strategies, players are bound to deterministic behavior, and the evader is quickly captured. (b)~Our approach lifts the strategy space to learn more competitive, mixed strategies, \ie,~distributions over multiple trajectory candidates per player.
    The opacity of each trajectory in (b) encodes probability of selecting that learned candidate.
    }
    \label{fig:equilibrium-differences}
\end{figure}

Perhaps more importantly, however, equilibrium solutions to trajectory games do not always exist.
Nonexistence arises even in extremely simple static games such as rock-paper-scissors, in which neither player wishes to commit to a fixed, deterministic action which could be exploited by its opponent.
Unsurprisingly, the same phenomenon can arise in more complex trajectory games.
For example, consider the game of tag shown in \cref{fig:equilibrium-differences}, where the red pursuer wishes to catch the blue evader.
Here, if the evader chooses a single, deterministic trajectory, it will certainly be caught by a rational pursuer.
In the context of small, discrete games such as rock-paper-scissors, these non-existence issues are commonly avoided by allowing players to ``mix'' their actions, i.e., to choose an action at random from a distribution of their choice. This distribution is called a ``mixed strategy,'' in contrast to the choice of a single deterministic action or ``pure strategy.''
However, for continuous trajectory games it can be difficult to represent mixed strategies.
Hence, it is common to regularize players' objectives in order to encourage the existence of pure solutions.\footnote{\looseness=-1 Regularizing players' control inputs to ensure the existence of equilibria is well-established in the literature on dynamic games and robust control \cite{basar1998gametheorybook}, \cite{basar2008h}.}
For example, in \cref{fig:equilibrium-differences-pure} each player is penalized for large accelerations, leading to an equilibrium in which the evader is cornered by the pursuer.

With these issues in mind, this paper introduces the following key contributions:
\begin{enumerate}
    \item a principled method for reducing the online computation needed to solve trajectory games via the introduction of an offline training phase, and
    \item a formulation of \emph{lifted} games over multiple trajectory candidates, which admit a natural class of high-performance mixed strategies. 
\end{enumerate}
Together, these contributions enable efficient and reliable on-line trajectory planning for autonomous agents in noncooperative settings, such as the tag example of \cref{fig:equilibrium-differences}.
We validate our methods in a suite of Monte Carlo studies, in which we demonstrate that lifting gives rise to mixed strategies as shown in \cref{fig:equilibrium-differences-mixed}, providing a significant competitive advantage in both open-loop and receding-horizon play.
Our method's reliable convergence and its ability to explicitly account for constraints enables training from scratch within only a few minutes of simulated self-play.
Once fully trained, learning can be disabled and our method generates mixed strategies within \SI{2}{\milli\second} for the tag example in \cref{fig:equilibrium-differences}.

\section{Related Work}
\label{sec:related}


Our contributions build upon recent work in trajectory optimization and game-theoretic planning, and bear a close relationship with work in learning motion primitives and implicit differentiation.
We discuss these relationships in further detail below.

\subsection{Trajectory Optimization}

Trajectory optimization refers to a finite-horizon optimal control problem in which a robot seeks a sequence of control inputs which minimize a performance criterion \cite{liberzon2011calculus}.
It is common to use trajectory optimization for model predictive control (MPC), whereby a robot quickly re-optimizes a new sequence of control inputs as new sensor data becomes available \cite{borrelli2017predictive}.
While a host of trajectory optimization techniques have been proposed in recent years, most common algorithms build upon the iterative linear-quadratic regulator \cite{li2004iterative,li2007iterative,todorov2005generalized} and differential dynamic programming \cite{jacobson1970differential,xie2017differential,tassa2014control}. 
In turn, these may be understood as specific approximations to standard algorithms in nonlinear programming (NLP), such as sequential quadratic programming \cite{nocedal2006optimizationbook,bertsekas1997nonlinear}.
As discussed below, this fundamental NLP representation underlies the proposed approach for multi-agent trajectory games.

\subsection{Trajectory Games}
\label{sec:background_traj_games}

Recent work has sought to generalize the aforementioned trajectory optimization techniques to address multi-agent, competitive planning.
Here, each player seeks to minimize an individual performance criterion subject to constraints arising from, e.g., dynamics and actuator limits. 
The objectives and constraints for different players may, in general, depend upon the trajectories of others.
Solutions to these problems are characterized by equilibrium points at which all players' strategies are unilaterally optimal.

The theoretical underpinnings of \emph{dynamic} games were established in the context of state feedback \cite{isaacs1954differential,basar1998gametheorybook,starr1969further,starr1969nonzero}. However, computational methods were historically limited to highly-structured problems such as those found in linear robust control \cite{basar2008h,green2012linear}.
Recent work on iterative linear-quadratic methods \cite{fridovich2020icra,laine2021arxiv,laine2021multi} extends these ideas to more general games such as those found in noncooperative robotic planning.

Closely related problems have also been studied in the context of \emph{static} games.
Here, equilibrium points are found by treating the trajectory of each player as a single action, and assuming the players choose these actions simultaneously~\cite{basar1998gametheorybook}.
This results in a Generalized Nash Equilibrium Problem (GNEP), for which general-purpose solution methods exist \cite{facchinei2009generalized,facchinei2010generalized,dirkse1995path}.
Several domain-specific solvers have been developed to exploit the structure of trajectory games, ranging from augmented Lagrangian \cite{cleac2020rss} to iterated best response methods~\cite{wang2019dars,wang2021game}.
Still, these methods can have a high computational burden in challenging settings.

Regardless of the equilibrium definition (dynamic or static), solving trajectory games is fundamentally harder than solving single-agent trajectory optimization problems, if for nothing else but the increased problem dimension. The number of decision variables scales linearly with the number of players involved, and even with proper handling of sparsity, computation generally scales cubicly with the number of players \cite{fridovich2020icra}. 
In \cref{sec:reduced-runtime}, we introduce an offline training phase for trajectory games which effectively reduces the online computational burden to that of solving a trajectory optimization problem for each player in parallel.

\subsection{Motion Primitives}

In this work, we introduce a \emph{trajectory lifting} technique, which may be understood in the context of motion primitives~\cite{khatib1997robot}.
As we discuss in \cref{sec:approach}, this reformulation endows each player with a distribution over finitely many trajectory candidates, which may be learned. 
However, learning trajectories, or motion primitives, is also meaningful in the context of a single agent, and recent work has proposed this concept in the contexts of quadcopter navigation \cite{camci2019learning} and robot manipulation \cite{stulp2011learning}.
In this light, the present paper may be viewed as a multi-agent generalization of these techniques.
Additionally, our work constitutes an adaptive, learning-enabled generalization of the multi-agent motion primitive games formulated for autonomous racing in \cite{liniger2019noncooperative}.
In \cref{sec:lifted}, we show that these trajectory primitives may be learned efficiently with first-order optimization. 

\subsection{Differentiable Optimization}

To improve learning efficiency, we employ implicit differentiation to propagate derivative information through all steps of our proposed trajectory lifting approach.
Recent works in end-to-end neural architectures for autonomous driving have developed specialized network layers that embed optimization problems \cite{amos2017optnet,agrawal2019differentiable,amos2019differentiable}.
Like these methods, 
we obtain derivatives of players' game values with respect to learnable parameters by implicitly differentiating through the first order optimality conditions for all players in a lifted trajectory game.
\section{Formulation}\label{sec:formulation}

We develop our approach in the context of games played between two\footnote{Although we limit our discussion to two players, our formulation may be extended to the general case. 
For further discussion, refer to \cref{subsec:many-players}. 
} agents over time, in which each agent's motion is characterized by a smooth discrete-time dynamical system.
That is, we model agent $i$'s motion as the temporal evolution of its state $\bar\bx_i(t) \in \mathbb{R}^n$ and control input $\bar \bu_i(t) \in \mathbb{R}^m$ over discrete time-steps $t \in \{1,\dots,T\}$ with $\bar\bx_i(t+1) = F\big(\bar\bx_i(t), \bar\bu_i(t)\big)$ for differentiable vector field $F(\cdot, \cdot)$.

Taking an egocentric approach, we investigate using \emph{model-predictive game play} (MPGP) \cite{fridovich2020icra,cleac2020rss,wang2019dars} as a method by which each player can plan strategically while accounting for the predicted reactions of its opponent.
MPGP constitutes a natural analogue to MPC \cite{borrelli2017predictive} for \emph{noncooperative}, multi-agent settings.
That is, at regular intervals the `ego' agent formulates a finite-horizon trajectory game between itself and its opponent. 
The equilibrium of this game specifies optimal trajectories for both players; the ego agent begins to execute its equilibrium trajectory, and the procedure repeats after a short time once the players have moved. 

The finite-horizon trajectory games formulated at each planning interval can be modeled as a pair of coupled optimization problems, as is common in the literature \cite{basar1998gametheorybook,laine2021arxiv}:
\begin{subequations}\label{eq:game_def_1}\begin{align}\begin{split}\label{eq:pursuer_1}
    \mathrm{OPT}_1(\tau_2, \bx_1) \ := \ \argmin_{\tau_1} &\ \ \ f_1(\tau_1, \tau_2) \\ 
    \text{s.t.} \ \  &\ \ \ \tau_1 \in \mathcal{K}_1(\bx_1)
\end{split}\\ \begin{split}\label{eq:evader_1}
    \mathrm{OPT}_2(\tau_1, \bx_2) \ := \ \argmin_{\tau_2} &\ \ \ f_2(\tau_1, \tau_2) \\ 
    \text{s.t.} \ \  &\ \ \ \tau_2 \in \mathcal{K}_2(\bx_2)
\end{split}\end{align}\end{subequations}
The decision variables $\tau_i$ for each player $i\in\{1,2\}$ represent discrete-time state-control trajectories starting from initial configuration $\bx_i$.
Therefore, the constraint sets $\mathcal{K}_i(\bx_i)$ represent the set of all trajectories satisfying dynamic constraints, control limits, etc.
Note that these sets need not be compact or convex, and that
the players' constraint sets are independent of one another's trajectory.
In contrast, the differentiable cost functions $f_i(\tau_1, \tau_2)$ in each problem can depend upon \emph{both} players' trajectories.
Thus, the $f_i$ can encode preferences such as goal-reaching and collision-avoidance.
In particular, since constraints are decoupled, we assume that any aspect of interaction in the game is modeled via the cost functions and not through constraints. 

As discussed in \cref{sec:background_traj_games}, existing methods to find local equilibrium solutions of \cref{eq:game_def_1} include iterative best response \cite{wang2019dars,wang2019game} and iterative linear-quadratic methods \cite{fridovich2020icra,laine2021arxiv,di2019cdc,cleac2020rss}. 
A Nash equilibrium for Game \cref{eq:game_def_1} starting from initial configuration $(\bx_1,\bx_2)$ is defined to be a pair of trajectories, $(\tau_1^*, \tau_2^*)$, satisfying
\begin{equation}
    \tau_1^* \in \mathrm{OPT}_1(\tau_2^*, \bx_1) \ \ \text{and} \ \  \tau_2^* \in \mathrm{OPT}_2(\tau_1^*, \bx_2).
\end{equation}
Nash equilibrium points encode rational strategic play for both players, and hence serve as a natural solution concept in trajectory games \cref{eq:game_def_1}. 
For this reason, most recent MPGP methods \cite{di2019cdc,cleac2020rss,fridovich2020icra,laine2021arxiv,wang2019game} aim to compute a  Nash equilibrium of the trajectory game \cref{eq:game_def_1}.
As a practical matter, however, Nash equilibria can be intractable to compute and modern methods often settle for \emph{local} equilibria, in which players' trajectories are only locally optimal.

Several important issues arise when employing an MPGP approach. The first is that solving for a Nash equilibrium---even a local Nash---is harder than solving for a locally optimal trajectory (as would be done in the single-agent setting of MPC). Not only is the search space larger due to the inclusion of both players' trajectory variables, but potential complications are also introduced by agents' different and potentially conflicting objectives. As in MPC, real-world applications depend upon our ability to compute solutions to \cref{eq:game_def_1} quickly; unfortunately though, this increased complexity can make MPGP unsuitable for real-time applications. 

The second issue is that a Nash equilibrium point may not even exist for Game \cref{eq:game_def_1}, particularly when one or both of the subproblems \cref{eq:pursuer_1} and \cref{eq:evader_1} are non-convex. 
Relatedly, even if a Nash equilibrium does exist, it may not be unique. Consequently, in MPGP an agent may spend significant computational effort searching for an equilibrium point that does not exist. Worse, non-uniqueness implies that even if an agent finds an equilibrium, the opponent's predicted Nash trajectory may not be representative of its true strategy. 

To make these issues more concrete, consider the following ``toy'' variant of the tag game in \cref{fig:equilibrium-differences}. 
Let $\tau_1$ and $\tau_2$ be scalars,  $f_1(\tau_1,\tau_2) = \|\tau_1-\tau_2\|_2^2 = -f_2(\tau_1,\tau_2)$, and $\mathcal{K}_1 = \mathcal{K}_2 = [-1,1]$. 
Here, the pursuer (Player 1) and evader (Player 2) choose positions in the interval $[-1,1]$.
By inspection, we may verify that no Nash equilibrium exists. 
With additional regularization, however, this example can be modified to admit \emph{local} equilibria. With $f_1$ defined as above, if we redefine the function $f_2(\tau_1,\tau_2) = -\|\tau_1-\tau_2\|_2^2-\|\tau_2\|_2^2$, two local equilibrium points result: $(-1,-1)$ and $(1,1)$. 
Unfortunately, the locality of these equilibria causes a significant problem: if Player~1 computed one of these equilibria, and Player~2 computed the other, the resulting pairing of actions, e.g. $(-1,1)$, would have a significantly different outcome for the players than what occurs at either local equilibrium. 



\section{Approach}
\label{sec:approach}

We propose a novel \emph{lifted} trajectory game formulation which ameliorates the complexity and existence/uniqueness issues discussed in \cref{sec:formulation}.

\subsection{Reducing Run-time Computation}
\label{sec:reduced-runtime}
To begin, we propose a technique for offloading the complexity introduced by multi-agent interactions to an offline training phase. The result of this pre-training is that at run-time, only a single-agent trajectory optimization problem remains for each player, and these problems can be solved in parallel. 
To do so, we introduce auxiliary trajectory references $\xi_i$ for each Player $i$, and with a slight abuse of notation, reformulate Game \cref{eq:game_def_1} as:
\begin{subequations}\label{eq:game_def_2}\begin{align}\begin{split}\label{eq:pursuer_2}
    \mathrm{OPT}_1(\xi_2, \bx_1, \bx_2) \ := \ \argmin_{\xi_1} &\ \ \ f_1(\tau_1, \tau_2) 
\end{split}\\ \begin{split}\label{eq:evader_2}
    \mathrm{OPT}_2(\xi_1, \bx_1, \bx_2) \ := \ \argmin_{\xi_2} &\ \ \ f_2(\tau_1, \tau_2) 
\end{split}\end{align}\end{subequations}

Here, the decision variables $\xi_i$ and the initial states $\bx_i$ determine trajectory variables $\tau_i = \mathrm{TRAJ}_i(\xi_i, \bx_i)$, which we presume to have the form:
\begin{equation}\label{eq:traj_def}\begin{aligned}
    \mathrm{TRAJ}_i(\xi_i, \bx_i) \ := \ \argmin_{\tau} &\ \frac{1}{2} \| G_i\tau - \xi_i \|_2^2 + \frac{1}{2}\| H_i\tau \|_2^2 \\ 
    \text{s.t.} \ \  &\ \ \ \tau \in \mathcal{K}_i(\bx_i).
\end{aligned}\end{equation}

The first term of the cost functions in problem \cref{eq:traj_def} enables $\xi_i$ to serve as a reference for $\tau_i$. For example, if $\tau_i = [X_i^\mathsf{T} \ U_i^\mathsf{T}]^\mathsf{T}$, with $X_i$ and $U_i$ representing the state and control variables of the trajectory, then $G_i$ could be $[0 \ I]$, giving $\xi_i$ the interpretation of a control reference signal. Alternatively, $\xi_i$ could represent a reference for the terminal state of the trajectory. 
The second term allows regularization of the trajectory, which may be needed if the reference and constraint sets are otherwise insufficient to isolate solutions.

In Appendix~\ref{appendix:equiv}, we prove that for any stationary point $(\tau_1,\tau_2)$ of Game \cref{eq:game_def_1}, there exists a stationary point $(\xi_1,\xi_2)$ of Game \cref{eq:game_def_2} such that for both players $i$, $\tau_i = \mathrm{TRAJ}_i(\xi_i,x_i)$. This implies that no stationary points are ``lost'' in the reformulation from \cref{eq:game_def_1} to \cref{eq:game_def_2}. Furthermore, we discuss practical methods to guarantee that all computed stationary points of \cref{eq:game_def_2} result in stationary points of \cref{eq:game_def_1}. This implies that no spurious stationary points are ``introduced'' in the reformulation.

With this reformulation, it is now possible to offload a significant amount of computation to an offline training phase. 
To do so, we propose training a \emph{reference generator} for each player, denoted by the function $\pi_{\theta_i}(\bx_1, \bx_2)$, which maps both player's initial states $(\bx_1, \bx_2)$ to reference $\xi_i$.
Generator $\pi_{\theta_i}$ is parameterized by $\theta_i$ and, e.g., may be a multi-layer perceptron as described in \cref{sec:implementation}. Given a data set\footnote{$D$ need not be constructed laboriously; in \cref{sec:selfplay-learning} we show that it can even be accumulated during online operation.} 
of initial MPGP configurations ${D := \{\bx_1^k,\bx_2^k\}_{k=1}^d}$, we train the reference generators $(\pi_{\theta_1}$ and $\pi_{\theta_2}$) by solving the following game offline:
\begin{subequations} \label{eq:game_def_training}
    \begin{align}\begin{split}
        \mathrm{GEN}_1(\theta_2, D) &:= \argmin_{\theta_1} \frac{1}{d}\sum_{k=1}^d f_1(\tau_1^k, \tau_2^k),
        \end{split}\\
        \begin{split}
        \mathrm{GEN}_2(\theta_1, D) &:= \argmin_{\theta_2} \frac{1}{d}\sum_{k=1}^d f_2(\tau_1^k, \tau_2^k).
        \end{split}
    \end{align}
\end{subequations}
Similar to Game \cref{eq:game_def_2}, each trajectory $\tau_i^k$ appearing in \cref{eq:game_def_training} is a function of $\theta_i,\bx_1^k$ and $\bx_2^k$, via the relationships
\begin{equation} \label{eq:tautraj}
\begin{aligned}
    \tau_1^k &= \mathrm{TRAJ}_1\Big(\pi_{\theta_1}(\bx_1^k,\bx_2^k),\bx_1^k\Big), \\
    \tau_2^k &= \mathrm{TRAJ}_2\Big(\pi_{\theta_2}(\bx_1^k,\bx_2^k),\bx_2^k\Big).
\end{aligned}
\end{equation}
A Nash equilibrium for Game \cref{eq:game_def_training} can be found by simultaneous gradient descent over each player's reference generator parameters, $\theta_1$ and $\theta_2$.
Simultaneous gradient play is widely used in adversarial machine learning, and is particularly important in both generative adversarial networks \cite{goodfellow2014arxiv} and multi-agent reinforcement learning \cite{foerster2017learning}.
Here, each player's parameter $\theta_i$ is iteratively updated as $\theta_i \leftarrow \theta_i - \delta \theta_i$, where
\begin{equation} \label{eq:training_update}
\begin{aligned}
    \delta\theta_1 &= \frac{\alpha_1}{d}\nabla_{\theta_1}\sum_{k=1}^d f_1\Big(\mathrm{TRAJ}_1(\theta_1,\bx^k), \mathrm{TRAJ}_2(\theta_2, \bx^k)\Big) \\ 
    \delta\theta_2 &= \frac{\alpha_2}{d}\nabla_{\theta_2}\sum_{k=1}^k f_2\Big(\mathrm{TRAJ}_1(\theta_1,\bx^k), \mathrm{TRAJ}_2(\theta_2, \bx^k)\Big)
\end{aligned}
\end{equation}
Note that in \cref{eq:training_update}, we use the shorthand
${\bx^k \equiv (\bx_1^k,\bx_2^k)}$, and although we abbreviate the arguments to the $\mathrm{TRAJ}$ functions, they should be interpreted exactly as in \cref{eq:tautraj}. The values $\alpha_1$ and $\alpha_2$ are learning rates used for the respective reference generators. 
To compute these gradients, we must differentiate through each player's objective $f_i$ and through each $\mathrm{TRAJ}_i$.
We have assumed \emph{a priori} that the $f_i$ were differentiable.
To differentiate through the trajectory optimization step of \cref{eq:traj_def}, we follow a procedure similar to what is outlined in \cite{amos2017optnet,agrawal2019differentiable,amos2019differentiable}.

Assuming that offline gradient play converges to a Nash equilibrium over the training set $D$, and that the resulting trajectory generators generalize to instances of \cref{eq:game_def_2} defined by configurations $(\bx_1,\bx_2)$ not included in $D$, then an approximate equilibrium solution to Game \cref{eq:game_def_1}, denoted by $(\hat{\tau}_1^*,\hat{\tau}_2^*)$ can be found via the following evaluations:
\begin{equation}
    \begin{alignedat}{3}
        \xi_1 &= \pi_{\theta_1}(\bx_1,\bx_2), 
        & \xi_2 &= \pi_{\theta_2}(\bx_1,\bx_2) \\
        \hat{\tau}_1^* &= \mathrm{TRAJ}_1(\xi_1, \bx_1), \ \ \ & \hat{\tau}_2^* &= \mathrm{TRAJ}_2(\xi_2, \bx_2)
    \end{alignedat}
\end{equation}

Hence, at run-time, solving this reformulated game only requires evaluating the reference generators and solving the optimization problems $\mathrm{TRAJ}_i$ to compute the corresponding trajectories.
These problems can be solved in parallel.
Furthermore, since trajectories are generated according to \cref{eq:traj_def}, each player's constraints defined by $\mc{K}_i(x_i)$ are guaranteed to be satisfied. 
Thus, if the reference generator does not generalize well, the only negative consequence is suboptimality (but not infeasibility).\footnote{Recall that each player's constraints do not depend upon the trajectory of the other player. Extension to this more complex case is possible, but beyond the scope of this paper.}

In summary, by pre-training a reference generator for each player offline, the run-time concerns of MPGP can be alleviated. Unfortunately, however, potential issues persist due to the possible non-existence or non-uniqueness of Nash equilibrium solutions. To address this concern, we introduce a concept we refer to as \emph{strategy lifting}.  

\subsection{Lifted Trajectory Games} \label{sec:lifted}

\begin{figure*}[t]
    \centering
    \includegraphics[width=0.9\textwidth]{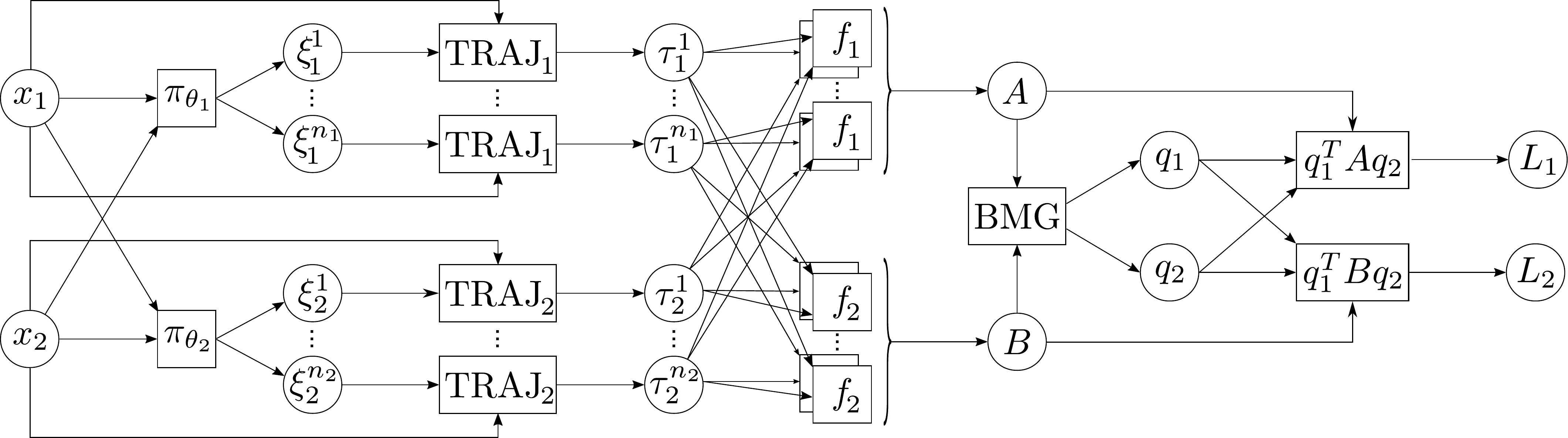}
    \caption{Overview of our proposed lifted game solver using reference generators. Generators  $\pi_{\theta_i}$ for both players are trained jointly to minimize their respective average losses $L_i$ (\ref{eq:l_def}).
At run-time deployment of this pipeline, the pre-trained generators produce references $\xi_i$ which approximate a Nash equilibrium of \cref{eq:game_def_3}. These references evaluate to equilibrium motion plan candidates $\tau_i$ and mixing strategies $q_i$ for each player. When used in an ego-centric MPGP fashion, e.g. for Player 1, $(\tau_1, q_1)$ serves as a distribution over ego motion plans, and $(\tau_2, q_2)$ constitutes a probabilistic opponent prediction.}
    \label{fig:pipeline_overview}
\end{figure*}

Rather than endowing each player with a single reference and its resulting trajectory, we allow each player to choose among multiple independent references according to the equilibrium solution of the bimatrix game formulated below: 
\begin{subequations}\label{eq:game_def_3}\begin{align}\begin{split}\label{eq:pursuer_3}
 \mathrm{OPT}_1^\mathrm{lifted}(\xi_2, \bx_1, \bx_2) :=  &\argmin_{\xi_1}   L_1(\xi_1,\xi_2)
\end{split}\\ \begin{split}\label{eq:evader_3}
  \mathrm{OPT}_2^\mathrm{lifted}(\xi_1, \bx_1, \bx_2)  :=  &\argmin_{\xi_2}  L_2(\xi_1,\xi_2)
\end{split}
\end{align}\end{subequations} 
where the dependence of $L_1$ and $L_2$ on $\xi_1$ and $\xi_2$ is made explicit through the following relationships:
\begin{subequations} \label{eq:pipeline_eqs}
\begin{align}
    \label{eq:l_def}
     L_1 &= q_1^\intercal A q_2, \ \ & L_2 &= q_1^\intercal B q_2 \\
     \label{eq:AB_def}
     A_{i,j} &= f_1(\tau_1^i, \tau_2^j), \ \ &  B_{i,j} &= f_2(\tau_1^i, \tau_2^j) 
     \end{align}
     \vspace{-22pt}
     \begin{align}
    \tau_1^i &= \mathrm{TRAJ}_1(\xi_1^i, \bx_1), \ i\in N_1\\
      \tau_2^j &= \mathrm{TRAJ}_2(\xi_2^j, \bx_2), \ j\in N_2 \\
      \label{eq:bmg_def}
     (q_1, q_2) &= \mathrm{BMG}(A, B).
\end{align}
\end{subequations}


Here, $N_1 := \{1,...,n_1\}$, and $N_2:=\{1,...,n_2\}$, where $n_1$ and $n_2$ are the number of trajectories for Player~1 and Player~2, respectively. Specifically, each variable $\tau_1^i$ represents one of $n_1$ trajectories that Player~1 optimizes over (and similar for Player~2).
The reference variables $\xi_i := (\xi_i^1, \dots, \xi_i^{n_i})$ are now collections of trajectory references, with each $\xi_i^j$ associated to~$\tau_i^j$.
The function $\mathrm{BMG}(A,B)$ maps cost matrices $A$ and $B$ to mixed equilibrium strategies for the resultant bimatrix game. Specifically, $\mathrm{BMG}(A,B)$ returns a point $(q_1^*, q_2^*) \in \mathbb{R}^{n_1}\times\mathbb{R}^{n_2}$ such that 
\begin{equation}\label{eq:bmg}
    \begin{alignedat}{3}
        (q_1^*)^\mathsf{T}Aq_2^* &\leq q_1^\mathsf{T}Aq_2^*, &&\forall q_1 \in \Delta^{n_1-1}, \\
        (q_1^*)^\mathsf{T}Bq_2^* &\leq (q_1^*)^\mathsf{T}Bq_2, \ \ \ &&\forall q_2 \in \Delta^{n_2-1}.
    \end{alignedat}
\end{equation}

In \cref{eq:bmg}, $\Delta^{k}$ is the $k$-simplex, representing the space of valid parameters for a categorical distribution over $k+1$ elements. Note that when $n_1=n_2=1$, Game \cref{eq:game_def_3} reduces exactly to Game \cref{eq:game_def_2}, since $\mathrm{BMG}(A,B) \equiv (1,1)$.

Continuous games, such as \cref{eq:game_def_1}, may suffer from non-existence of equilibrium points, but when those games are \emph{separable}, a mixed strategy equilibrium is known to exist with finite support \cite{stein2008separable,glicksberg20169}. This theoretical result motivates the lifting \cref{eq:game_def_3} of our reference-based formulation of Game \cref{eq:game_def_1}. 

For comparative purposes, Game \cref{eq:game_def_3} is presented analogously to \cref{eq:game_def_2}, i.e. without any explicit dependence on reference generators. Nevertheless, generators can be trained analogously to \cref{eq:game_def_training}, using a similar simultaneous gradient procedure. As before, solving \cref{eq:game_def_3} or an analogous version of \cref{eq:game_def_training} via gradient play requires that each of the function evaluations in \cref{eq:pipeline_eqs} are differentiable in their arguments. It has already been discussed how each of these functions are differentiable, with the exception of the bimatrix game in \cref{eq:bmg_def}. We discuss in Appendix \ref{appendix:bmg} how this function is also differentiable.


A summary of the lifted game solver that utilizes reference generators for reduced online computation is provided in \cref{fig:pipeline_overview}.
With this computation graph, the cost of approximating solutions to Game \cref{eq:game_def_3} is that of evaluating the two generator calls, solving the resultant $n_1 + n_2$ trajectory optimization problems (in parallel, if warranted), and solving a bimatrix game formed by considering all combinations of player trajectories. 



\subsection{Extension to Many-Player Games}\label{subsec:many-players}
We reiterate that, although we present this formulation in the two player setting, generalizations to larger games are straightforward. In this case, each player would consider multiple trajectory candidates, and a cost \emph{tensor} would be created for each player, representing the costs for all possible combinations of players' trajectories. A finite Nash equilibrium could be identified over these cost tensors to compute the equilibrium mixing weights $q_i$ 
\cite{papadimitriou2005computing}, and computed by solving a nonlinear, mixed complementarity program \cite{Laine_TensorGames, dirkse1995path}.
Note that the majority of computation required to construct these cost tensors can be trivially parallelized, making our framework particularly promising for many-player settings.
We defer further study of such games to future work.

\subsection{Implementation}
\label{sec:implementation}

We implement the lifted game solver depicted in \cref{fig:pipeline_overview} in the Julia programming language \cite{bezanson2017sirev}.
For the experiments conducted in this work, reference generators $\pi_{\theta_i}$ are realized as multi-layer perceptrons, trajectory optimization problems $\textrm{TRAJ}_i$ are solved via OSQP \cite{stellato2020osqp}, and bimatrix games are solved using a custom implementation of the Lemke-Howson algorithm \cite{lemke1964equilibrium}.

In order to facilitate back-propagation of gradients through this computation graph, we utilize the auto-differentiation tool Zygote \cite{innes2018zygote}.
For those components that cannot be efficiently differentiated automatically, namely $\textrm{TRAJ}_i$ and $\textrm{BMG}$ in \cref{fig:pipeline_overview}, we provide custom gradient rules via the implicit function theorem, \cf \cite{agrawal2019differentiable,amos2019differentiable,amos2017optnet} and Appendix~\ref{appendix:bmg}.
Our implementation can be found at 
\url{https://lasse-peters.net/pub/lifted-games}.
\section{Results}
\label{sec:results}

We have presented a novel formulation of lifted trajectory games in which learned reference generators facilitate the efficient online computation of mixed strategies.
In this section, we evaluate the performance of our proposed lifted game solver on variants of the ``tag'' game shown in \cref{fig:equilibrium-differences} and described below in \cref{sec:setup_tag}.
Concretely, we aim to quantify the utility of learning trajectory references rather than choosing them \emph{a priori} (\cref{sec:importance_learning}), characterize the equilibria identified by trajectory lifting (\cref{sec:equilibrium-differences}), evaluate the performance of trajectory lifting in head-to-head decentralized competition (\cref{sec:baseline-competition}), and demonstrate our method's capacity for online training in receding horizon MPGP (\cref{sec:selfplay-learning}). Our supplementary material includes a video summarizing these results.

\subsection{Environment: The Tag Game}
\label{sec:setup_tag}

We validate our methods in a two-player tag game, illustrated in \cref{fig:equilibrium-differences}.
 Here, each player's trajectory $\tau_i$ follows time-discretized planar double-integrator dynamics $\ddot p_i = u_i$, where $p_i\in\R^2$ is understood to represent horizontal and vertical position in the plane.
The set $\mc{K}_i(x_i)$ then encompasses all dynamically-feasible trajectories that also satisfy input saturation limits and state constraints.
In particular, we require that positions remain within a closed set, such as the pentagon illustrated in \cref{fig:equilibrium-differences}, and that speeds remain below a fixed magnitude.
These choices yield \emph{linear} constraints, so that \cref{eq:traj_def} becomes a quadratic program.
We note, however, that our approach does not rely upon this convenient structure and is compatible with more general embedded nonlinear programs.

For the purposes of this example, we shall designate Player~1 to be the ``pursuer'' and Player~2 to be the ``evader.''
Hence, the pursuer's objective $f_1(\tau_1, \tau_2)$ measures the average distance between players' trajectories over time and is regularized by the difference in control effort between the two players to ensure the existence of at least local pure Nash equilibria for the original game \cref{eq:game_def_1}.
The evader's objective is $f_2(\tau_1, \tau_2) = -f_1(\tau_1, \tau_2)$.  
Since the tag game has zero-sum cost structure, throughout the following evaluations we only report the cost for the pursuer and refer to this quantity as the \emph{game value}.
Furthermore, unless otherwise stated, we use an input reference signal $\xi_i$ for all players in \cref{eq:game_def_2} and \cref{eq:game_def_3}.

\subsection{The Importance of Learning Trajectory Candidates}
\label{sec:importance_learning}

Without lifting, it is still possible to approximate mixed strategies for the trajectory game by discretizing the trajectory space (e.g., via sampling \cite{liniger2019noncooperative}).
We compare to a sampling-based mixed-strategy baseline to study the isolated effects of learning in a lifted space.

\noindent \textbf{Setup.} 
We instantiate an evader with $n_2 = 20$ pre-sampled trajectory references.
To strengthen the evader, we ensure that these samples cover a large region of the trajectory space.
To that end, in this experiment (only) we use $\xi_i$ as a reference for Player $i$'s goal state rather than their input sequence. 
We compare the pursuer's performance for two different schemes of generating trajectory candidates.
The non-learning baseline \emph{samples} $n_1 \in \{1,\dots,20\}$ pursuer trajectory references from the same distribution as the evader.
Our method computes the pursuer strategy by performing gradient play on \cref{eq:pursuer_3} to \emph{learn} 2 trajectory candidates via the goal reference parameterization.
The mixed Nash equilibrium $(q_1, q_2)$ over the players' trajectory candidates is computed according to \cref{eq:pipeline_eqs}.
We evaluate both methods for 50 random initial conditions, and record the game value for each trial.

\noindent \textbf{Discussion.}
\Cref{fig:discretization_eval} summarizes the results of this experiment.
As shown, the baseline steadily improves its performance with increasing numbers of sampled trajectory references to mix over.
However, even with 20 trajectory samples, it cannot match the performance of our approach with only two learned candidates.
Moreover, learning only a few trajectory references drastically reduces the number of trajectory optimizations and, consequently, the size of the bimatrix game in~\cref{eq:pipeline_eqs}.

\begin{figure}[H]
    \centering
    \includegraphics[scale=\makiescale]{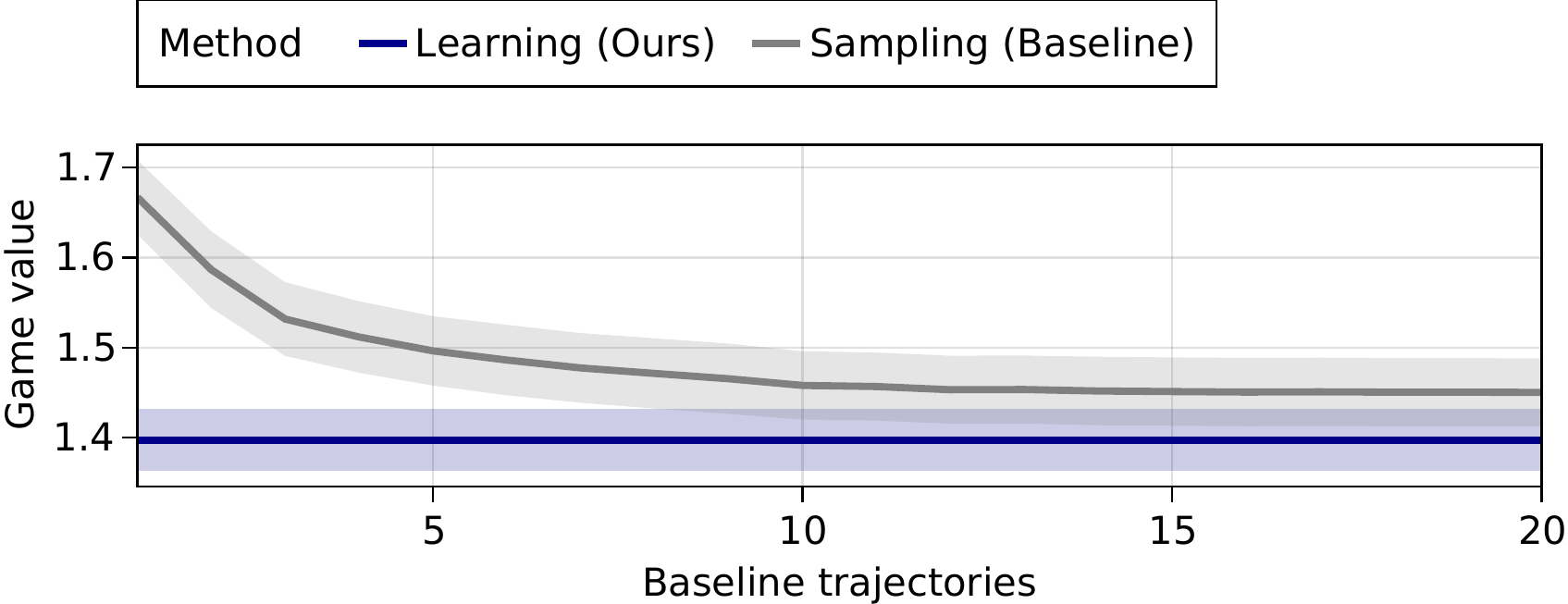}
    \caption{
    Comparison of game value for both sampled and learned pursuer trajectories. 
    Lines trace the sample mean over 50 randomized trials, and the surrounding ribbons denote the SEM.
    On the horizontal axis, we vary the number of sampled \emph{baseline} trajectories $n_1$, while fixing \emph{our} approach to learn only 2 trajectories.
    }
    \label{fig:discretization_eval}
\end{figure}

\subsection{Convergence and Characteristics of Lifted Equilibria}
\label{sec:equilibrium-differences}

In this experiment, we analyze \emph{mixed} strategies found by our lifted solver and compare them to \emph{pure} strategies computed by a non-lifting baseline.
We shall demonstrate that both approaches reliably converge to different equilibria, and characterize these differences. 

\noindent \textbf{Setup.}
We perform a Monte Carlo study in which we randomly sample 20 initial states of the tag game.
On each sample, we invoke two solvers which perform gradient play on different strategy spaces.
The baseline solver is restricted to pure strategies as in Game \cref{eq:game_def_2}.\footnote{Such pure Nash solutions could also be found using iterated best response~\cite{wang2019dars}, iterative linear-quadratic methods \cite{di2019cdc}, or mixed complementarity methods \cite{dirkse1995path}.}
Our method utilizes lifting to find mixed strategies which solve Game \cref{eq:game_def_3}.
In each iteration of gradient play, we record the game value.

\noindent \textbf{Discussion.} \cref{fig:equilibrium_learning} shows the reliable convergence of both methods in this Monte Carlo study.
Since both players learn competitively via simultaneous gradient play, the game value ought not to evolve monotonically; an equilibrium is reached when neither player can improve its strategy unilaterally.
At convergence, we observe that the mixed strategies found by our lifting procedure result in a higher game value.
This higher value implies that, by operating in a lifted strategy space, the evader can secure a greater average distance between itself and the pursuer.

\begin{figure}
    \centering
    \includegraphics[scale=\makiescale]{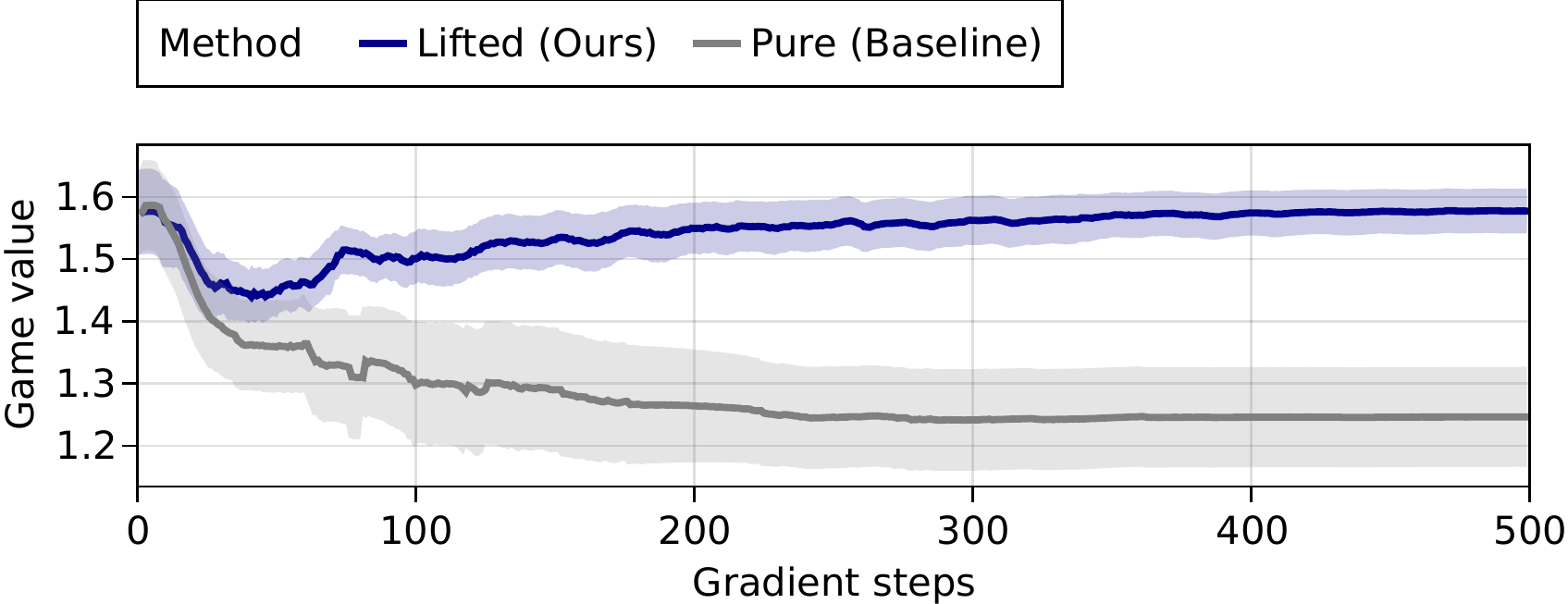}
    \caption{
    Equilibrium value convergence averaged over 20 initial states.
    Ribbons indicate the SEM.
    The baseline approximates Nash equilibria in \emph{pure} strategies via gradient play on \cref{eq:game_def_2}.
    Our method approximates Nash equilibria in \emph{lifted} strategies via gradient play on \cref{eq:game_def_3}.
    \label{fig:equilibrium_learning}
    }
\end{figure}

This gap in value may be understood intuitively by examining the strategy profiles for each method shown in \cref{fig:equilibrium-differences}.
In \cref{fig:equilibrium-differences-pure}, players are restricted to pure strategies, and a rational pursuer can exploit the evader's deterministic choice of trajectory.
In contrast, our proposed lifting formulation allows the evader to mix over multiple trajectory candidates, \cf \cref{fig:equilibrium-differences-mixed}, making its motion less predictable and hence increasing the chance of escaping the pursuer.
In response, the pursuer also mixes between two trajectory candidates.
However, each of the pursuer's candidates must account for the full distribution of evader trajectories; hence, the pursuer plans to turn less aggressively than the evader.

In this experiment, we have studied a \emph{centralized} setting in which each method computes strategies for both players from a single game.
Therefore, the results presented above are only suitable to characterize the solution points of \cref{eq:game_def_2} and \cref{eq:game_def_3}, but do not justify conclusions about the competitive performance of these solutions in decentralized settings, such as MPGP.
In the next section, we extend our analysis to settings in which the opponent's decision-making process is unknown.

\subsection{Competitive Evaluation Against Non-Lifted Strategies}
\label{sec:baseline-competition}

This experiment is designed to examine the performance of both pure (Baseline) and lifted (Ours) strategies in \emph{decentralized} head-to-head competition.
For this purpose, we perform two additional Monte Carlo studies which simulate tournaments among players in each strategy class.

Note that, in contrast to previous experiments, here, player strategies are not computed as the solution to a single, centralized game.
Rather, each player is oblivious to their opponent's decision making process and solves its own version of the game from a known initial state over a finite time interval. 

\subsubsection{Open-Loop Competition}
\label{sec:open-loop-competition}
To begin, we evaluate both methods in open-loop on a fixed, 20-step time horizon.

\noindent\textbf{Setup.}
For this Monte Carlo study, we randomly sample 100 initial states.
For every sampled state, we invoke pure and lifted game solvers twice with randomly sampled initial strategies; once to obtain pursuer strategies, and once to obtain evader strategies.\footnote{This initialization procedure avoids leaking information about players' decision making processes to one another.}
For all possible solver pairings on these 100 state samples we record the resultant value of the competing strategies; i.e., if Player~$i$ chooses trajectory $\tau_i$, we record $f_1(\tau_1, \tau_2)$.

\noindent\textbf{Discussion.}
\Cref{tab:open_loop_comparison} summarizes the mean and the standard error of the mean (SEM) of the resultant game value for this open-loop tournament.
The evader has a clear incentive to utilize lifted strategies, since they secure the highest game value irrespective of the solution technique used by the pursuer.
The best response of the pursuer is then also to play a lifted strategy to minimize value within this column.
Hence, 
(Ours, Ours) is the unique Nash equilibrium in this meta game between solvers.

Additionally, observe that the baseline pursuer performs very well against the baseline evader, as deterministic evasion strategies can always be exploited by a rational pursuer.
However, the tournament value reported in the bottom right of \cref{tab:open_loop_comparison} is inconsistent with the equilibrium value for the baseline found earlier in \cref{fig:equilibrium_learning}.
This discrepancy suggests that players in this decentralized setup find different local solutions depending on the initialization of the baseline solver. 
Hence, random initialization effectively makes even a pure strategy evader slightly unpredictable, thereby allowing it to attain a higher average value.
By contrast, the value of the lifted strategy computed by our method (top left, \cref{tab:open_loop_comparison}) closely agrees with the equilibrium value computed in \cref{fig:equilibrium_learning}, which indicates that non-uniqueness of solutions is not an issue for our approach.\footnote{This close agreement in value suggests that our method identifies global (rather than local) Nash equilibria which satisfy the so-called \emph{ordered interchangeability property} \cite{basar1998gametheorybook}. Unfortunately, as in continuous optimization, it is generally intractable to properly verify that these solutions are global.}

\begin{table}
    \centering
    \caption{\label{tab:open_loop_comparison} Open-loop competition.}
    \begin{tabular}{ccc}
        \toprule
                  & \multicolumn{2}{c}{\textbf{Evader}} \\
        \cmidrule(r){2-3}
        \textbf{Pursuer} & Lifted & Pure \\
        \midrule
         Lifted          & $1.577 \pm 0.021$ & $1.502 \pm 0.022$      \\
         Pure            & $1.672 \pm 0.022$ & $1.370 \pm 0.027$        \\
         \bottomrule
    \end{tabular}
\end{table}
\begin{table}
    \centering
    \caption{\label{tab:receding_horizon_comparison} Receding-horizon competition.}
    \begin{tabular}{ccc}
        \toprule
                  & \multicolumn{2}{c}{\textbf{Evader}} \\
        \cmidrule(r){2-3}
        \textbf{Pursuer} & Lifted & Pure \\
        \midrule
         Lifted          & $1.360 \pm 0.003$ & $1.289 \pm 0.005$      \\
         Pure            & $1.463 \pm 0.004$ & $0.903 \pm 0.009$      \\
         \bottomrule
    \end{tabular}
\end{table}

\begin{figure*}
    \centering
    \subfigure[Closed-loop game value\label{fig:selfplay-value}]{
    \hspace{-18pt}
    \includegraphics[scale=\makiescale]{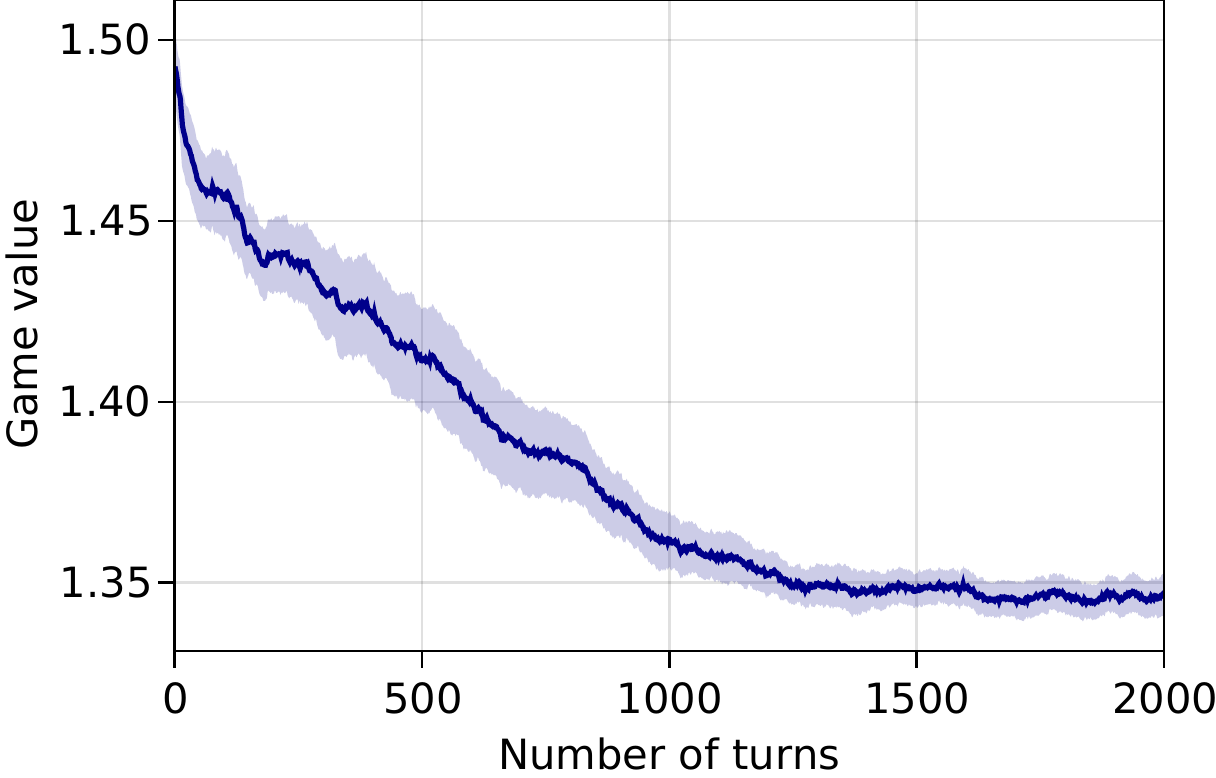}
    }
    \hspace{7.5pt}
    \subfigure[Initial strategies\label{fig:selfplay-initial-strategy}]{
    \hspace{-25pt}
    \includegraphics[scale=\makiepngscale]{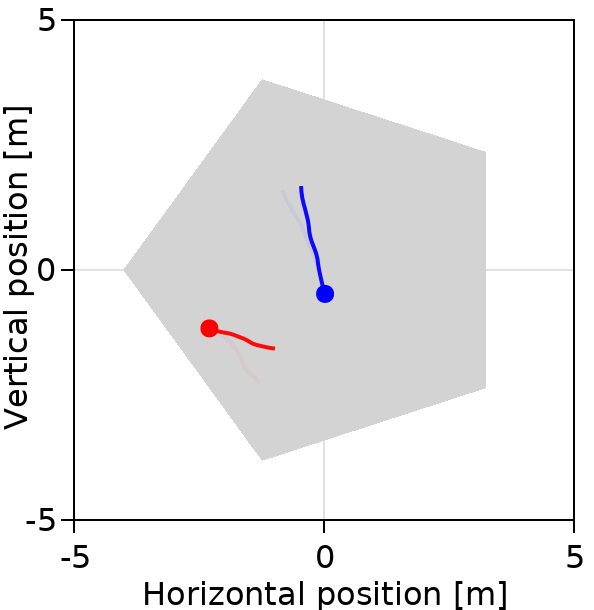}
    }
    \hspace{8pt}
    \subfigure[Intermediate strategies\label{fig:selfplay-intermediate-strategy}]{
    \hspace{-25pt}
    \includegraphics[scale=\makiepngscale]{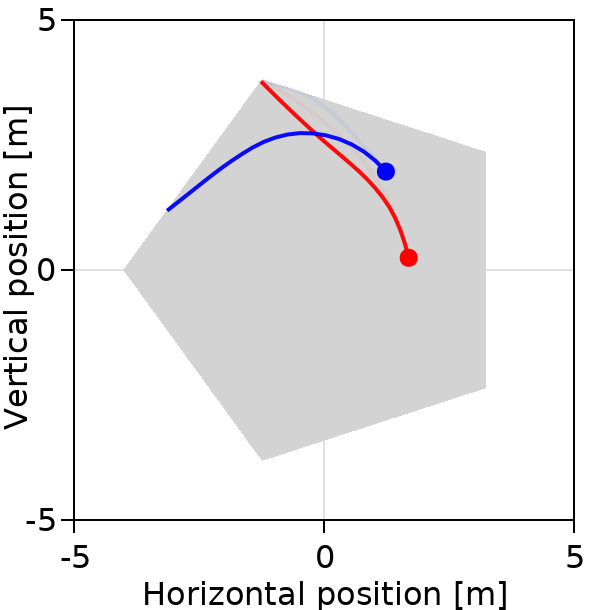}
    }
    \hspace{8pt}
    \subfigure[Final strategies\label{fig:selfplay-final-strategy}]{
    \hspace{-25pt}
    \includegraphics[scale=\makiepngscale]{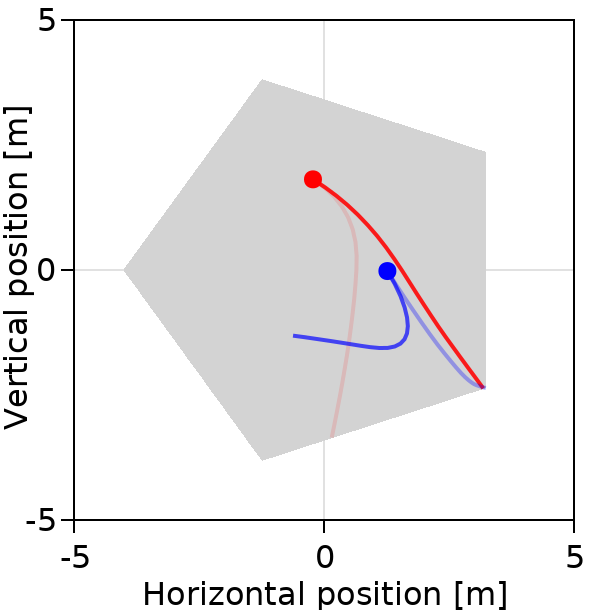}
    }
    \caption{
    Training a reference generator for the lifted game solver in \cref{fig:pipeline_overview} in simulated self-play. Transparency denotes the mixing probabilities associated with each trajectory which arise from $\mathrm{BMG}$. (a) Game value mean and SEM over a moving window of 500 turns. (b-d) Strategies at different phases of learning.
    }
    \label{fig:selfplay}
\end{figure*}

\subsubsection{Receding-Horizon Competition}
\label{sec:receding-horizon-competition}

As MPGP is naturally applied in a receding horizon fashion, we replicate the previous Monte Carlo study in that setting.

\noindent\textbf{Setup.}
For each of 5 state samples, we simulate receding-horizon competitions for all possible solver pairings.
As before, we use a planning horizon of 20 time steps for all players, and in order to simulate latency, we only allow players to update their plans every 9 time steps.
Each simulation terminates once players have updated their strategy for the $500^{\mathrm{th}}$ time.
From one such trial, we compute the game value by evaluating the pursuer's objective $f_1$ on the entire \emph{closed-loop} trajectories of both players.
Note that, in contrast to previous experiments, here we use pre-trained reference generators for all solvers as described in \cref{sec:reduced-runtime} to accelerate the computation in this large simulation.

\noindent\textbf{Discussion.}
\Cref{tab:receding_horizon_comparison} summarizes both the mean and SEM for the resultant game value in this receding-horizon Monte Carlo tournament.
Overall, we observe the same patterns as in the open-loop setting: lifting is the dominant strategy for the evader and the corresponding best response for the pursuer.
However, the game values found for this receding horizon setting are generally lower than in open-loop.
By replanning in receding horizon, the pursuer can react to the evader's decision before the distance between them grows very large.

\subsection{Learning in Receding-Horizon Self-Play}

Finally, we demonstrate that a lifted game solver with trajectory generators, as shown in \cref{fig:pipeline_overview}, can be rapidly trained from scratch in simulated self-play.

\noindent\textbf{Setup.}
We repeat the following experiment 10 times.
For each player, we randomly sample an initial state $\bx_i$ and initialize their reference generator $\pi_{\theta_i}$ with parameters $\theta_i$ sampled from a uniform distribution.
Subsequently, we simulate receding-horizon learning over 2500 turns with the lifted game solver in the loop.
In contrast to the setup used in \cref{sec:receding-horizon-competition}, here, we do \emph{not} use pre-trained reference generators.
Instead, the network parameters are updated on the fly using gradient descent.
That is, at every turn, we first perform a forward pass through the computation graph of \cref{fig:pipeline_overview} to compute a lifted strategy profile, followed by a backwards pass to compute a gradient step on each player's reference generator parameters~$\theta_i$.
For each experiment, we record players' strategies as well as the game value over a moving window of 500 turns. 

\noindent\textbf{Discussion.}
\cref{fig:selfplay} summarizes the results of lifted learning in self-play.
Initially, the untrained reference generators cause both players to move haphazardly, \cf \cref{fig:selfplay-initial-strategy}.
As learning progresses, players become more competitive, resulting in purposeful, dynamic maneuvers, \cf \cref{fig:selfplay-intermediate-strategy}.
Within approximately 1500 turns, learning converges, the game value stabilizes, and the solver has learned to generate highly competitive mixed strategies as shown in \cref{fig:selfplay-final-strategy}.

Note that, throughout the learning procedure, state and input constraints are explicitly enforced in the $\mathrm{TRAJ}$ step of the pipeline in \cref{fig:pipeline_overview}.
Moreover, since our proposed pipeline is end-to-end differentiable, it provides a strong learning signal.
Therefore, training in simulated self-play over 2500 turns can be performed in less than three minutes on a standard laptop.
Then, once the reference generators $\pi_{\theta_i}$ are fully trained, learning can be disabled, and a forward pass on the pipeline in \cref{fig:pipeline_overview} can be computed with an average run-time of \SI{2}{\milli\second}.
In summary, these results indicate that our method learns quickly and reliably, making it well-suited for online learning in real systems with embedded computational hardware.

\label{sec:selfplay-learning}

\section{Conclusion}
\label{sec:conclusion}

In this paper, we have proposed two key contributions to the field of noncooperative, multi-agent motion planning.
First, we have introduced a principled technique to reduce the online computational complexity of solving these trajectory games.
Second, we extended this approach to optimize over a richer, probabilistic class of \emph{lifted} strategies for each player.
Taken together, these innovations facilitate efficiently-computable, high-performance online trajectory planning for multiple autonomous agents in competitive settings.
Moreover, our method directly accounts for problem constraints and hence guarantees that learned trajectories satisfy these constraints whenever they are feasible.

While our formulations readily extend to games with many players and arbitrary cost structure, we demonstrate our results in a two-player, zero-sum game of tag.
We validate our approach in extensive Monte Carlo studies, in which we observe rapid and reliable convergence to solutions which outperform those which emerge in the original, non-lifted strategy space.

Finally, we showcase our approach in online learning, where each player solves lifted trajectory games in a receding time horizon.
Despite the additional complexity present in this setting---e.g., non-stationary training data and potential limit cycles---our 
method converges reliably to competitive mixed strategies.
These initial results are extremely encouraging, and future work should investigate online learning and adaptation in noncooperative settings more extensively.
In particular, we note that our method is limited to so-called \emph{open-loop} information structures, in which each agent in a trajectory game must choose future control inputs as a function only of the current state.
We believe that the incorporation of feedback structures at the trajectory-level will be an exciting direction for future research.

\section*{Acknowledgments}
This work was supported in part by the National Police of the Netherlands.
All content represents the opinion of the authors, which is not necessarily shared or endorsed by their respective employers and/or sponsors.
L. Ferranti received support from the Dutch Science Foundation NWO-TTW within the Veni project HARMONIA (nr. 18165).

\bibliographystyle{plainnat}
\bibliography{glorified,new}

\appendices

\section{Equivalence of Game \cref{eq:game_def_1} and Game \cref{eq:game_def_2}} \label{appendix:equiv}

In this section we establish an equivalence result between \cref{eq:game_def_1} and \cref{eq:game_def_2}. We prove this for a particular interpretation of the  reference variables $\xi_i$, and forms of $G_i, H_i, \mathcal{K}_i(x_i)$, noting that similar results can be established for other settings.

For each Player $i$, consider the instance of \cref{eq:traj_def} in which $G_i := \mathrm{I}$ and $H_i := 0$, representing the identity and zero matrices of appropriate dimension. Furthermore, assume that the set $\mathcal{K}_i(x_i) := \{ \tau : lb_i \leq g_i(\tau) \leq ub_i \}$, for some vector-valued and twice-differentiable function $g_i$, and lower and upper bounds $lb_i$ and $ub_i$. It is assumed that a suitable constraint qualification applies to this constraint set, such as the LICQ~\cite{nocedal2006optimizationbook}. This implies that $\xi_i$ has dimension equal to that of the decision variable $\tau$, and the objective of \cref{eq:traj_def} is to find a trajectory  $\tau\in\mathcal{K}_i(x_i)$ which is as close as possible to $\xi_i$ as measured by the $\ell_2$-norm.

\begin{theorem} \label{thm}
In the setting as stated above, \begin{enumerate}
    \item For any stationary point $(\tau_1,\tau_2)$ of Game \cref{eq:game_def_1}, there exists a stationary point $(\xi_1,\xi_2)$ of Game \cref{eq:game_def_2} such that $\tau_i = \mathrm{TRAJ}_i(\xi_i,x_i)$ for all players $i$.
    \item For any stationary point $(\xi_1,\xi_2)$ of Game \cref{eq:game_def_2} satisfying $\xi_i\in\mathcal{K}_i(x_i)$, the trajectories $\tau_i = \mathrm{TRAJ}_i(\xi_i,x_i)$ constitute a stationary point for \cref{eq:game_def_1}. 
\end{enumerate}
\end{theorem}

\begin{proof}
To prove this result, we first make explicit the definition of of a stationary point for \cref{eq:game_def_1} and \cref{eq:game_def_2}. A stationary point for \cref{eq:game_def_1} is a point $(\tau_1,\tau_2)$, such that for both players $i$,
\begin{equation} 
\label{eq:nash_cond_0}
d^\mathsf{T}\nabla_{\tau_i}f_i(\tau_1,\tau_2) \geq 0, \forall d \in T_{\mathcal{K}_i}(\tau_i).
\end{equation}

Here, $T_{\mathcal{K}_i}(\tau_i)$ is the set of linearized feasible directions with respect to constraint set $\mathcal{K}_i(x_i)$ at $\tau_i$, which because we have assumed a suitable constraint qualification, is equivalent to the tangent cone at this point \cite{nocedal2006optimizationbook}. Specifically, at a feasible point $\tau$,  let $\mathcal{I}_l(\tau) := \{j: g_{i,j}(\tau) = lb_j\}$, and $\mathcal{I}_u(\tau) := \{j: g_{i,j}(\tau) = ub_j\}$. Then

\begin{equation}\label{eq:tangent_cone}
\begin{aligned}
T_{\mathcal{K}_i}(\tau) := \{ d : \  &d^\mathsf{T}\nabla g_{i,j}(\tau) \geq 0, j\in\mathcal{I}_l(\tau), \\
&d^\mathsf{T} \nabla g_{i,j}(\tau) \leq 0, j\in\mathcal{I}_u(\tau)\}
\end{aligned}
\end{equation}

A stationary point for \cref{eq:game_def_2} is a point $(\xi_1, \xi_2)$ such that
\begin{equation}
    \label{eq:nash_cond_1}
    ( \nabla_{\xi_i}\tau_i \cdot  d)^\mathsf{T}\nabla_{\tau_i}f_i(\tau_1,\tau_2) \geq 0, \forall d,
\end{equation}
where $\tau_i = \mathrm{TRAJ}_i(\xi_i, x_i)$. Note that $( \nabla_{\xi_i}\tau_i \cdot  d) := (\nabla_{\xi_i}\mathrm{TRAJ}_i(\xi_i, x_i) d)$ as appearing above is the directional derivative of $\mathrm{TRAJ}_i(\xi_i,x_i)$ with respect to changes of $\xi$ in the direction $d$. This directional derivative is defined to be $e$, where $e$ solves the following quadratic program \cite{ralph1995directional}:

\begin{equation} \label{eq:directional_qp}
    \begin{aligned}
        &\min_{e} \ \ \  \frac{1}{2} e^\mathsf{T} Q_1 e + d^\mathsf{T}Q_2 e \\
        &\text{s.t.} \ \ \ \  e\in\mathcal{C}_{\lambda_i}(\tau_i),
    \end{aligned}
\end{equation}
where $Q_1 := I - \nabla_{\tau,\tau}^2 (g(\tau_i)^\mathsf{T} \lambda_i)$, $Q_2 := -I$, $\lambda_i$ are the dual variables associated with the primal solution $\tau_i$ to $\mathrm{TRAJ}_i(\xi_i,x_i)$, and $\mathcal{C}_{\lambda_i}(\tau)$ is the critical cone to the constraint set $g_i(\tau)$ with respect to $\lambda_i$ at $\tau$:

\begin{equation}\label{eq:critical_cone}
\begin{aligned}
C_{\lambda_i}(\tau) := \{ d : \  \bar{lb}_{i,j} \leq d^\mathsf{T} \nabla g_{i,j}(\tau) \leq \bar{ub}_{i,j})\}.
\end{aligned}
\end{equation}
The bounds $\bar{lb}_{i,j}$ and $\bar{ub}_{i,j}$ are defined as:
\begin{equation}
        (\bar{lb}_{i,j},\bar{ub}_{i,j}) := \begin{cases}
            (0,0) & j\in\mathcal{I}_l(\tau) \ \& \ \lambda_{i,j} > 0\\
            (0,\infty) & j\in\mathcal{I}_l(\tau) \ \& \ \lambda_{i,j} = 0 \\
            (-\infty,\infty) & i\notin(\mathcal{I}_l(\tau)\cup\mathcal{I}_u(\tau))  \\
            (-\infty,0) & j\in\mathcal{I}_u(\tau) \ \& \ \lambda_{i,j} = 0 \\
            (0,0) & j\in\mathcal{I}_u(\tau) \ \& \ \lambda_{i,j} < 0 \\
        \end{cases}
\end{equation}




Now, to prove 1), we show that $\xi_i = \tau_i$ satisfies the claim. It follows directly that $\tau_i = \mathrm{TRAJ}_i(\xi_i, x_i)$. It can be verified that, because $\tau_i\in\mathcal{K}_i(x_i)$ by definition, then all constraints appearing in $\mathrm{TRAJ}_i$ are only weakly active, implying $\lambda_i~=~0$. This implies that the constraint set appearing in \cref{eq:directional_qp} is precisely the tangent cone \cref{eq:tangent_cone}. Therefore, for all directions $d$, $\nabla_{\xi_i}\mathrm{TRAJ}_i(\tau_i, x_i) \cdot d \in T_{\mathcal{K}_i}(\tau_i)$, which by \cref{eq:nash_cond_0}, implies that \cref{eq:nash_cond_1} holds, establishing the result. 

To prove 2), we simply note that if $\xi_i \in \mathcal{K}_i(x_i)$, then $\tau_i=\xi_i$. Furthermore, in this setting $\lambda_i=0$ as before, and therefore the critical cone appearing in \cref{eq:directional_qp} is again equivalent to the tangent cone \cref{eq:tangent_cone}.  This implies that the directional derivative $(\nabla_{\xi_i} \tau_i \cdot d)$ is defined to simply be the projection of the direction $d$ into the tangent cone at $\tau_i$. The set of all directions $d$ mapped through this projection results precisely in $T_{\mathcal{K}_i}(\tau_i)$. Therefore, the conditions \cref{eq:nash_cond_1} imply \cref{eq:nash_cond_0} for this setting, implying our result. 
\end{proof}

The result as stated in \cref{thm} does not imply that an \emph{arbitrary} stationary point found for \cref{eq:game_def_2} corresponds to a stationary point for \cref{eq:game_def_1}, since it may be that either of the references $\xi_i \notin \mathcal{K}_i(x_i)$. For such reference points,  it is possible that for some direction $d$ the expression in \cref{eq:nash_cond_1} holds with equality, yet the expression in \cref{eq:nash_cond_0} is violated. This situation results in ``sticky constraints,'' in which a descent direction exists for $f_i(\tau_1,\tau_2)$, yet that direction is not in the range of $\nabla_{\xi_i}\mathrm{TRAJ}_i(\xi_i,x_i)$, i.e. small changes to the reference are not enough to release $\tau_i$ away from the active constraint boundaries. 

To address this issue, we propose a modest regularization scheme to eliminate the possibility of reference stationary points of \cref{eq:game_def_2} which do not correspond to trajectory stationary points of \cref{eq:game_def_1}. One such approach could be to enforce constraints in \cref{eq:game_def_2} such that $\xi_i\in\mathcal{K}_i(x_i)$. This, however, would render the reformulation from \cref{eq:game_def_1} to \cref{eq:game_def_2} pointless. Instead, we impose a simple regularization in the objectives of each player in \cref{eq:game_def_2}. Namely, instead of minimizing over $f_i(\tau_1,\tau_2)$ w.r.t. $\xi_i$, we minimize over 
\begin{equation}\label{eq:regularization}
    f_i(\tau_1,\tau_2) + \|(g(\xi_i)-ub)_+ + (lb-g(\xi_i))_+ \|_2^2,
\end{equation}
where $(\cdot)_+ := \max(\cdot, 0)$.

Note that this introduced regularization is exact, and has precisely the effect of eliminating any stationary points for \cref{eq:game_def_2} in which $\xi_i\notin \mathcal{K}_i(x_i)$. If the regularization term is non-zero, then necessarily from the definition of the directional derivative \cref{eq:directional_qp}, the gradient of the regularization component is in the null-space of $\nabla_{\xi_i}\mathrm{TRAJ}_i(\xi_i,x_i)$. This implies the regularization can be driven to zero without changing the resultant solution $\tau_i$. This is true irrespective of the scale factor multiplying the regularization term. Furthermore, if $\xi_i\in\mathcal{K}_i(x_i)$, then the regularization term is zero, and has no effect on stationary points of the un-regularized game \cref{eq:game_def_2}.

We note that the particular choice of  regularization \cref{eq:regularization} is only applicable for the interpretation of the references $\xi_i$ made throughout this section. For more general parameterizations of the reference, as discussed in the main text, a suitable regularization is the norm of inequality constraint multipliers associated with the solution of $\mathrm{TRAJ}_i(\xi_i,x_i)$. The use of this dual-variable regularization is effective at eliminating the spurious stationary points for \cref{eq:game_def_2}, so long as the parameterization of the reference is rich enough such that for any $\xi_i$ and associated $\tau_i,\lambda_i$, there exists directions $d$ in which the $\xi_i$ can be perturbed and the directional derivative of $\tau_i$ is $0$, and the directional derivative of $\lambda_{i,j}$ is negative for all $j$. This is true, for example, of the control signal reference used throughout this work. 

Therefore, with use of the introduced regularization \cref{eq:regularization}, the stationary points of Games \cref{eq:game_def_1} and \cref{eq:game_def_2} have a one-to-one correspondence, warranting the use of Game \cref{eq:game_def_2} in place of Game \cref{eq:game_def_1}.

\section{Differentiating through $\mathrm{BMG}$ } \label{appendix:bmg}

The problem of finding $q_1,q_2$ which satisfy (\ref{eq:bmg}) (as is the task of the function $\mathrm{BMG}$), can be equivalently expressed as the linear complementarity problem \cite{murty1988linear} 
\begin{equation} \label{eq:lcp}
    \begin{aligned}
        \text{find} \ \  &p_1, p_2 \\
        \text{s.t.} \ \ \ p_1 \geq 0 &\perp \bar{A}p_2 \geq 1\\
        p_2 \geq 0 &\perp \bar{B}^\mathsf{T}p_1 \geq 1.
    \end{aligned}
\end{equation}
The solution $(q_1,q_2)$ to the BMG are related to the solution to \cref{eq:lcp} by the relations \begin{equation} \label{eq:q_p}
\begin{aligned}
    (q_1)_i &= \frac{(p_1)_i}{ \sum_k (p_1)_k}, & 
    (q_2)_i &= \frac{(p_2)_i}{ \sum_k (p_2)_k}.
\end{aligned}
\end{equation}
It is assumed that $\bar{A}$ and $\bar{B}$ are derived from the original matrices $A,B$, as the following. $\bar{A}_{i,j} := A_{i,j} + \alpha$, $\bar{B}_{i,j} := B_{i,j} + \beta$, for some positive constants $\alpha, \beta$ such that every element of $\bar{A}$ and $\bar{B}$ are strictly positive. Furthermore, the $1$s appearing in the right-hand side of the constraints in \cref{eq:lcp} are assumed to represent vectors of appropriate dimension with each value equal to $1$.

Consider some solution $p_1,p_2$ to $(\ref{eq:lcp})$ in which strict complementarity holds for each condition, e.g. either $p_{1,j} = 0$ or $(\bar{A}p_{2})_j = 1$, but not both. For each $j\in\{1,2\}$, Denote the index sets $\mathcal{I}_j^+ := \{i : (p_j)_i > 0\}$. Then let \begin{align*}
    p_1^+ &:= [p_1]_{\mathcal{I}_1^+}, & p_2^+ &:= [p_2]_{\mathcal{I}_2^+}, \\
    \bar{A}^+ &:= [A]_{\mathcal{I}_1^+, \mathcal{I}_2^+}, & \bar{B}^+ &:= [B]_{\mathcal{I}_1^+, \mathcal{I}_2^+}.
\end{align*}
In words, $p_1^+$ is the vector formed by only considering the non-zero elements of $p_1$, and $\bar{B}^+$ is the matrix formed by considering the columns specified by $\mathcal{I}_1^+$ and rows specified by $\mathcal{I}_2^+$. By the strict complementarity, at equilibrium, it is that
\begin{equation} \label{eq:positive_system_bmg}
    \begin{bmatrix} 
        0 & \bar{A}^+ \\ (\bar{B}^+)^\intercal & 0
    \end{bmatrix} \begin{bmatrix} p_1^+ \\ p_2^+ \end{bmatrix} = 1,
\end{equation}
where, as before, the right-hand side $1$ is a vector consisting of all $1$s.

The values $p_j^- := [p_j]_i, i\notin \mathcal{I}_j^+$ are defined to be identically~0, and as such have 0 derivative with respect to the values~$\bar{A}, \bar{B}$. The derivatives of remaining portion of the solution,~$p_1^+, p_2^+$, can be evaluated from \cref{eq:positive_system_bmg}. If the matrix on the left-hand-side of \cref{eq:positive_system_bmg} is singular, then the resulting solution is in fact non-isolated (there exist a continuum of solutions satisfying \cref{eq:lcp}), and the derivatives of the solution are not defined. If the matrix is non-singular, then necessarily so are $\bar{A}^+$ and $\bar{B}^+$, and the isolated solutions of $p_1^+,p_2^+$ are locally related to the matrices $\bar{A}$, $\bar{B}$ as
\begin{equation}\label{eq:pos_solutions}
\begin{aligned}
    p_1^+ &= (\bar{B}^+)^{-\intercal} 1, \\
    p_2^+ &= (\bar{A}^+)^{-1} 1.
\end{aligned}
\end{equation}

In this form, the derivatives of each element of $p_j^+$ can be found by differentiating through the expressions \cref{eq:pos_solutions}. Combining the above results, the derivatives of the solution vector $p_1, p_2$ with respect to the problem data $A$, $B$, can be established as the following:
\begin{equation} \label{eq:q_derivs}
\begin{aligned}
    &\frac{\partial (p_1)_i}{\partial A_{j,k}} := 0, \\
    &\frac{\partial (p_1)_i}{\partial B_{j,k}} := \begin{cases} 0 &: i \notin \mathcal{I}_1^+ \\
    -((\bar{B}^+)^{-\intercal} I_{k,j} p_1^+)_i &: \text{else}
    \end{cases}, \\
    &\frac{\partial (p_2)_i}{\partial A_{j,k}} := \begin{cases} 0 &: i \notin \mathcal{I}_2^+ \\
    -((\bar{A}^+)^{-1} I_{j,k} p_2^+)_i &: \text{else}
    \end{cases}, \\
    &\frac{\partial (p_2)_i}{\partial B_{j,k}} := 0.
\end{aligned}
\end{equation}

Above, the term $I_{j,k}$ is used to refer to the matrix consisting of zero everywhere except at the $(j,k)$-th position, which has value $1$. 

When strict complementarity does not hold at the solution to \cref{eq:lcp}, then only directional derivatives of the solution vectors exist w.r.t. the problem data. The various directional derivatives are found by, for each condition which does not hold with strict complementarity, making a selection on whether that index should be included the sets $\mathcal{I}_j^+$ or not. Then proceeding with the remainder of calculations, the result forms one of the directional derivative for the system. The directions for which this derivative is valid are defined to be those which make the directional derivative consistent with the selected index sets. 

The derivatives of the elements of~$p_1$ and~$p_2$ with respect to the cost matrices are formed by propagating the derivatives~\cref{eq:q_derivs} through the relationships~\cref{eq:q_p}.
\end{document}